\documentclass[runningheads]{llncs}

\usepackage{ifthen}
\def\Context{}

\def\Context{WITHCOMMENTS}
\ifthenelse{\equal{\Context}{WITHCOMMENTS}}{
    \newcommand{\comment}[1]{\textcolor{red}{#1}}

}
{
    \newcommand{\comment}[1]{}
}

\usepackage[english]{babel}
\usepackage{tikz}
\usetikzlibrary{arrows,patterns}
\usepackage{mathtools}
\usepackage{tcolorbox}
\usepackage{lipsum}
\usepackage{paralist}
\usepackage{multirow}
\usepackage{chngpage}
\usepackage{fullwidth}
\usepackage{afterpage}
\usepackage{cellspace, multirow, booktabs}
\usepackage{adjustbox}
\usepackage{diagbox}
\usepackage{placeins}

\usepackage[framemethod=tikz]{mdframed}

\mdfdefinestyle{MyFrame}{
    linecolor=black,
    innertopmargin=\baselineskip,
    innerbottommargin=\baselineskip,
    backgroundcolor=white}

\newcommand{\Distances}[1]{}

\usepackage{algorithm, caption}
\usepackage[noend]{algpseudocode}
\usepackage{amsfonts}
\usepackage{amssymb}
\algrenewcommand\alglinenumber[1]{\scriptsize #1:} 
\usepackage{enumerate}
\usepackage{mathrsfs}
\usepackage{multirow}
\usepackage{wasysym}

\usepackage{graphicx}
\usepackage[colorinlistoftodos]{todonotes}
\usepackage[utf8]{inputenc}
\usepackage{todonotes}
\usepackage[colorlinks=true, allcolors=blue]{hyperref}
\usepackage{forest}
\usepackage{subfigure}
\usepackage{url}
\usepackage{pgfplots}
\usepackage{alltt}
\usepackage{boxedminipage}
\usepackage{color}
\usepackage{siunitx} 
\usepackage{enumitem}

\makeatother

\makeatletter
\let\@@pmod\pmod
\DeclareRobustCommand{\pmod}{\@ifstar\@pmods\@@pmod}
\def\@pmods#1{\mkern4mu({\operator@font mod}\mkern 6mu#1)}
\makeatother

\DeclareMathAlphabet{\mathpzc}{OT1}{pzc}{m}{it}

\usepackage[title]{appendix} 

\usepackage[utf8]{inputenc}
\begin{document}
\title{Rational Threshold Cryptosystems}
\author{David Yakira \and Ido Grayevsky \and Avi Asayag}
\institute{Orbs Research\\
\email{$\{$david,ido,avi$\}$@orbs.com}\\
}
\maketitle

\section{Abstract}
\label{Abstract}

We propose a framework for threshold cryptosystems under a permissionless-economic model in which the participants are rational profit-maximizing entities. To date, threshold cryptosystems have  been considered under permissioned settings with a limited adversary. Our framework relies on an escrow service that slashes and redistributes deposits to incentivize participants to adhere desired behaviors. Today, more than ever, sophisticated escrow services can be implemented over public blockchains like Ethereum, without additional trust assumptions. The key threat to rational threshold cryptosystems is \emph{collusion}---by cooperating ``illegally'', a subset of participants can reveal the cryptosystem's secret, which, in turn is translated to unfair profit. Our countermeasure to collusion is \emph{framing}. If the escrow is notified of collusion, it rewards the framer and slashes the deposits of all other participants. We show that colluding parties find themselves in the prisoner's dilemma, where the dominant strategy is framing.




\section{Introduction}
\label{Introduction}

In threshold cryptosystems several parties cooperate in order to produce a signature, decrypt a ciphertext, or reveal a secret. In particular, in a $(t,n)$-threshold cryptosystem there are $n$ participants, any $t+1$ of which need to cooperate in order to use the system. Any smaller subset is doomed to fail.

The motivation for our work stems from the key insight that threshold cryptosystems have an inherent conflict---on the one hand they require the cooperation of their participants in order to function, while on the other hand, in certain circumstances, cooperation completely undermines the functionality of the system and is undesired. To our knowledge, this difficulty has not been explicitly formulated and a proper solution has not been proposed.

A simple example is found in $(t,n)$ publicly verifiable secret sharing~\cite{PVSS}, where $t+1$ players can co-decrypt a player's commitment. The players are instructed to co-decrypt only after all players have shared their commitments, and not beforehand. Otherwise, the output of the protocol can be easily manipulated (\cite{Ouroboros} use randomness beacons in this spirit). 

Consider another picturesque example. Satoshi Nakamoto wishes to reveal her identity, but only after she dies. She is willing to pay $R$ bitcoins to achieve her goal. She anonymously contacts $n$ known figures in the Bitcoin community and offers them a deal---she would share with them, via Shamir Secret Sharing, a signed Bitcoin transaction transferring each of them $R/n$ bitcoins, from an address known to belong to Satoshi. The secret would include her real-world identity, such that the cooperation of $t+1$ of them would both reveal her identity and allow each participant to realize her $R/n$ reward. Of course, from the participants' perspective, there is no reason to trust this anonymous character, let alone believe she is Satoshi. All the same, the promised $R/n$ bitcoins should incentivize them to participate\footnote{Proving that the transaction indeed encompasses the promised bitcoins, without revealing its input UtxOs, can be done in zero-knowledge.}. The problem is that while Satoshi only wants to reveal her secret after she dies
, the participants have clear motivation to cooperate before that time and gain an immediate reward (and finally find out who Satoshi is).

Traditional threshold cryptosystems assume that the participants follow the protocol blindly, unless they have been corrupted by the \emph{adversary}, which then controls their behavior. The adversary is limited and can corrupt up to $t<n/2$ participants\footnote{The restriction $t<n/2$ guarantees honest majority, which implies that there are $t+1$ honest participants that would advance the protocol. When robustness is not of concern, this restriction can be relaxed.}. This model immediately circumvents the key difficulty mentioned above---when cooperation is not allowed by the protocol, non-corrupted participants will simply refuse to cooperate. Evidently though, an adversary controlling $t$ participants is unlikely to stop there as she is highly motivated to corrupt the $t+1$ participant and get control of the underlying cryptosystem.

To tackle the problem we must first define formally what can be gained and what are the stakes in the system. We therefore start by considering an economic layer atop the cryptosystem. We abandon the assumption that participants can be trusted, and assume instead that they are rational entities seeking to maximize their own profit within this economic environment. To be able to shape their behavior we rely on an escrow service that takes deposits and follows a specific set of rules. Such an escrow can be realized, for example, as a smart contract over Ethereum~\cite{EthereumButerin2013whitepaper}. 

The idea of enhancing a cryptosystem using an escrow service was suggested by Randao~\cite{RandaoPaperV085}---a multi-party coin tossing protocol for generating randomness. There, if a participant fails to submit her valid secret, her deposit is slashed. There are three main caveats to this approach. First, the price of failing the protocol is low---failing amounts to convincing a \emph{single} participant to refuse to submit her secret. 
Second, it remains cryptic as to what is the incentive of participants to join the service in the first place. Implicitly, one can assume there is some intrinsic value tied to the system, but this is not articulated and Randao does not provide an analysis tying this value to the required deposits. 
Finally, in its current implementation, Randao invokes too much interaction with Ethereum, and cannot scale well with the number of participants\footnote{This could be improved using off-chain communication and interactive dispute protocols in the spirit of App.~\ref{Escrow-DKG-Ethereum}.}.

A rational threshold cryptosystem that has been widely studied is \emph{rational secret sharing}~\cite{RationalSecretSharingHalpern,RationalSecretSharingGordon,RationalSecretSharingLindell,RationalSecretSharingNonSimultanous1,RationalSecretSharingNonSimultanous2}. There however, only a specific aspect of the problem is dealt with, namely that participants prefer to learn the secret \emph{alone}. The objective of the protocol is to guarantee that all participants learn of the secret together. In particular, cooperation at unintended times is not an option under their model and is not considered.


Our framework handles both problems---cooperating ``illegally'', which we refer to as \emph{collusion}; and the ability to assure legitimate cooperation when it is required, which we refer to as \emph{robustness}. For a given threshold cryptosystem, our framework allows to assess the economic constraints on collusion (Sec.~\ref{collusion}) and the price to interrupt with the robustness of the protocol (Sec.~\ref{Robustness}). It thus allows system designers to build systems with adequate economic collusion resistance and economic robustness, as we show in Sec.~\ref{properties}. It also provides tools for users to assess the trustworthiness of the system.

We demonstrate the usefulness of our framework with two examples. In the first, unique threshold signatures are used to realize a randomness beacon (as done in~\cite{Dfinity,RandaoPaperV085} for example), which is then used for a simple lottery game. We also elaborate on Satoshi's example, previously mentioned.


\section{Background---Threshold Cryptography}

\emph{Shamir's secret sharing} (SSS)~\cite{SecretSharing} is an algorithm that realizes a threshold cryptosystem in which $t+1$ participants need to cooperate in order to reveal a secret. The procedure is simple. The owner of a secret $x$ randomly generates a degree $t$ polynomial $f$, where $f(0)=x$. She then numerates the participants $P_1,\dots,P_n$ and shares with $P_i$ the value $f(i)$, which is referred to as $P_i$'s \emph{secret key share}. Now, any $t+1$ participants can reconstruct the secret $f(0)$ by using Lagrange interpolation. No smaller set of participants can deduce any information on $f(0)$ whatsoever.

Threshold cryptosystems for signatures or encryption consist of two stages:
\begin{enumerate}
\item A key generation procedure, executed once to set up the system by generating the necessary keys---participants end up with secret key shares, that together correspond to some global secret key $x$; $x$ is kept secret from \emph{all} participants, but $X$, its corresponding public key, is known to all.   
\item An application stage, where the key shares may be used jointly over and over in order to sign messages or co-decrypt ciphertexts. 
\end{enumerate}
Since the key shares must be related to one another, key generation must be coordinated and cannot be done independently by each participant. In the absence of a trusted coordinator, \emph{Distributed Key Generation} (DKG) protocols allow a set of $n$ participants to jointly distribute key shares, without leaking any information on the secret $x$, and while maintaining the desired threshold $t$.

\label{Ped-DKG}
Our work relies on a well-known family of DKG protocols for discrete-log based threshold cryptosystems~\cite{FeldmanVSS,PedersenDKG,DKG,DKGRevisited} that rely on SSS. Feldman~\cite{FeldmanVSS} showed how SSS can be made \emph{verifiable}, by having the protocol produce public data allowing the participants to check the validity of their private key shares. Feldman's \emph{verifiable secret sharing} (VSS) 
relies on a trusted dealer to generate the global secret. Pedersen~\cite{PedersenDKG} was first to eliminate the need for a trusted dealer by parallelizing $n$ Feldman's VSS runs, each led by a participant $P_i$ committing (publicly) to her \emph{local} polynomial $f_i$. Then, the global polynomial $f$ is defined to be the sum of the local polynomials, and the global secret is $f(0)$. This protocol is known as Ped-DKG. In App.~\ref{Ped-DKG-Appendix} we give an overview of this protocol.  

In~\cite{DKG,DKGRevisited}, Gennaro et al. formulate a minimal set of requirements for a $(t,n)$-DKG protocol in discrete-log based cryptosystems over a group $G=\langle g \rangle$ of prime order $q$: 

\paragraph{Correctness.} 
\begin{itemize}
\item[(C1)] There is an efficient procedure that on an input consisting of any $t+1$ distinct (and correct) secret key shares produced by the protocol outputs the same secret $x$. There is another efficient procedure\footnote{This procedure alludes to some efficient way to verify the validity of a single share.} that on an input consisting of $n$ different secret key shares (of which at least $t+1$ are correct) and the public data produced by the protocol, outputs the secret $x$.
\item[(C2)] There is an efficient procedure to calculate $g^x$ from the public data produced by the protocol, where $x$ is the unique secret key guaranteed by (C1). 

\item[(C3)] $x$ is uniformly distributed in $Z_q$, hence $g^x$ is uniformly distributed in $G$.
\end{itemize}

\paragraph{Secrecy.} No information on $x$ can be learned by the adversary except for what is implied by the value $g^x$. More formally, secrecy depends on the existence of a distribution, parameterized by $g^x$ alone, which generates data that is indistinguishable from the public data produced by the protocol (in the eyes of the adversary, which is assumed to know $t$ secret key shares). 

\section{Rational Threshold Cryptosystems}
\label{Protocol}
In this section we set the ground for threshold cryptosystems that guarantee correctness, secrecy, and robustness in rational environments. 
We propose a framework that encompasses the entire lifetime of the cryptosystem---from key generation to application. The key generation kick-starts the system by distributing \emph{correct} keys to the participants. These are later used within some desired application, in which \emph{secrecy} and \emph{robustness} are to be maintained. 


\subsection{Model}

In our model a set of $n$ rational \emph{participants} $P_1,\dots,P_n$ take part in the cryptosystem. We cannot expect them to follow some predetermined set of instructions, and assume instead that they are driven to maximize their \emph{utility}. We make a distinction between participants and \emph{entities}---different entities have independent utilities, but different participants may correspond to the same entity\footnote{One  approach to model \emph{permissioned participation} is by assuming that participants correspond to different entities. Our model is permissionless in this sense.}. The utility of an entity is her total profit (or loss) by the end of the system's lifetime. 

We assume the existence of 
a transparent escrow service that takes participants' deposits and can burn deposits, or redistribute them, according to predetermined rules. External entities, i.e., ones that do not have a participant in the cryptosystem (we refer to them as the \emph{public}), are exposed to the cryptosystem via the transparent escrow service. 
Entities communicate with the escrow via (signed) transactions. The escrow is assumed to be publicly available with a known bound on transactions' processing time. Such an escrow can be implemented as a smart contract over Ethereum\footnote{Ethereum, however, has very limited throughput and the cost for on-chain transactions is high. While many techniques minimize the on-chain component of the system and carry all of the heavy lifting off-chain (e.g., TrueBit~\cite{TrueBit}, and Arbitrum~\cite{Arbitrum}), in this section we focus on the economic perspective of our system, and ignore limitations imposed by Ethereum. In a sense, we assume ``an ideal Ethereum'' whose capacity and utilization costs resemble that of a modern computer. In App.~\ref{Escrow-DKG-Ethereum} we show how to adjust our framework to comply with Ethereum in its current state.}. Moreover, we assume a public broadcast channel and a complete network of private channels between the participants (as required by Ped-DKG). Finally, we do not force any constraints on the ratio between $t$ and $n$---different ratios reflect different trade-offs. 



\subsubsection*{Economic adversarial model.}

To formulate our economic model we assume that a threshold cryptosystem holds some total value $R$. During its lifetime, $\alpha R$ of that value, for some $0\leq \alpha \leq 1$, is intended to reward the participants. The reward however, is conditioned on the  cooperation of the participants under application-specific limitations. In expectancy, each participant gets $\frac{\alpha R}{n}$. The remaining, $(1-\alpha) R$ goes to the external entities. If $t+1$ participants manage to cooperate ``illegally'' while ignoring the limitations, they can enjoy a larger portion of the reward $R$\footnote{This model suits the economic setting in Bitcoin---$R$ is $21$ million bitcoins, $\alpha=1$, $t=\frac{n}{2}$, and one participant can be thought of as one hash rate unit. Selfish mining~\cite{SelfishMining,SelfishMining-AvivZohar} is interpreted as ``illegal'' collusion, where $51\%$ of the hash rate results in $100\%$ of the rewards. Of course, the total hash rate in Bitcoin is ever changing and miners join and leave the network as they please, which means maintaining $51\%$ of the hash rate can be difficult.}. We refer to such cooperation as \emph{collusion}. The main goal of our framework is to prevent the option of profitable collusion.   



\paragraph{Example I.}
In a simple lottery game the winning ticket is determined according to a random value in the form of a unique threshold signature over some nothing-up-my-sleeve (publicly known) value. The total money raised selling lottery tickets is $R$, where $(1-\alpha)R$ is the total prize and $\alpha R$ is paid to the participants generating the signature. The lottery must guarantee that no entity becomes aware of the winning ticket whilst the ticket sale is taking place\footnote{For instance, if the escrow is a smart contract over Ethereum, the schedule of the lottery would be determined by Ethereum block heights---certain block heights for the ticket sale and later blocks for producing the signature.}. Thus, the application-specific limitation is that the signature is produced only after the ticket sale closes.

\paragraph{Example II.}
In Satoshi's example, $R$ is the total amount of bitcoins paid by Satoshi, and $\alpha=1$. Collusion translates to immediate profit rather than having to wait until Satoshi passes away.

\subsection{Design rationale}

Our economic model strongly captures the inherent collusion pressure in threshold cryptosystems. The main goal of the framework we propose is to dismantle this pressure by aligning the participants' utility-maximizing strategies with the application-specific limitations. The system designer's only tool in shaping the behavior of the participants is the escrow service, which (unlike the participants) follows  prescribed rules in the protocol without deviations.  

Our first design choice is to condition participation on a deposit made to the escrow. Fundamentally, we interpret a threshold cryptosystem as an investment channel in which entities invest an initial sum of money (the deposit), and hope for a return (the reward, $\frac{\alpha R}{n}$)\footnote{In our Bitcoin analogy, this is equivalent to miners buying ASICs and paying electricity bills for a potential reward.}. As in any investment channel though, a risk aspect is inevitable\footnote{While in most investment channels the risk stems from market forces, here, there are two types of risks: technical risks (e.g., software bugs or cyberattacks) and behavioral risks, which depend on the choices the other participants make.}. Indeed, the escrow can decide to slash deposits in case \emph{complaints} are filed. 

Complaints are the mechanism by which the escrow guides the participants' behavior in practice. The escrow supports a set of complaints that should deter participants from deviating from their intended behavior. The complaint mechanism is composed of four elements: 
\begin{enumerate}
\item Detection. A participant detects unintended behavior by another participant.
\item Proof. She provides the escrow with evidence of her findings.
\item Arbitration. The escrow verifies her proof.
\item Incentive layer. The escrow burns deposits and rewards participants according to some predetermined specification.   
\end{enumerate}

A complaint mechanism is measured by its ability to support detection and efficient arbitration of as many unintended behaviors as possible. Moreover, its deterrence relies on complaining being profitable, which should be addressed by the incentive layer. 

Complaints may have additional consequences atop redistributing and burning deposits. In the original Ped-DKG protocol complaints result in revealing secret sub-shares\footnote{Revealing up to $t$ of her sub-shares, a participant does not leak any information about her local polynomial. This is of course theoretically true, but weakens the security guarantees of the protocol. Imagine a situation where participant $P_1$ has revealed $t$ of her sub-shares (due to $t$ complaints against her). Without loss of generality, let those be $f_1(2), \dots, f_1(t+1)$. Now, in order to reveal $f_1$, it suffices to hack \textbf{any} of the participants that did not complain, i.e., $P_{t+2}, \dots, P_n$. This is a much easier task then having to specifically hack $P_1$. We thus avoid instructing participants to reveal their sub-shares.} (or excluding participants from the set QUAL, see App.~\ref{Ped-DKG-Appendix}), but under the assumed adversarial model there, the DKG procedure never fails. In contrast, our framework must allow the protocol to fail, for example when collusion was detected. This introduces an attack vector against the system which is to be mitigated by setting an appropriate cost to failing the system. We refer to this cost as \emph{economic robustness}. 

\subsection{Escrow-DKG and application}

\afterpage{
 
\noindent\makebox[\textwidth][c]{
\begin{minipage}{1\linewidth}

\begin{mdframed} [font=\small, align=left,
innertopmargin=4pt,
    innerbottommargin=4pt,
    innerrightmargin=4pt,
    innerleftmargin=4pt,
    skipabove=0pt,
    skipbelow=0pt,]
    \captionof{figure}{\textbf{Escrow-DKG}}

\label{Escrow-DKG-fig}
\begin{enumerate}

\item[1] Enrollment. The relevant participants enroll to the protocol by submitting \emph{enrollment} transactions to the escrow service. 
We refer to the participants by their enrollment order $1\leq i \leq n$. To enroll, participant $P_i$ locally generates:
\begin{enumerate}
\item A random polynomial of degree $t$ over $Z_q$, $$f_i(z)=a_{i,0}+a_{i,1}z+\dots+a_{i,t}z^t$$

\item A private key $\xi_i \in Z_q$ for private communication.

\item Public commitments to $f_i$'s coefficients, $X_{i,k}=g^{a_{i,k}}$ for $k \in \{ 0,\dots, t \}$.
\end{enumerate}

She then submits her enrollment transaction, 
$$tx_1(i)=\Big[\Delta,\gamma^{\xi_i}, H(X_{i,0}) \Big]$$
where, $\Delta$ is her deposit, $\gamma^{\xi_i}$ is her public encryption key and $H(X_{i,0})$ is her hash commitment to her \emph{partial secret} $a_{i,0}$.

\item[2] Public commitments. Every participant $P_i$ publishes \emph{commitments} transaction, $$tx_2(i)=\big[\{ X_{i,k} \}_{k=0}^t \big]$$

If $P_i$ fails to publish $tx_2(i)$ in time, any participant $j$ can file a complaint, $cm_1(j,i)$, against $P_i$. If $X_{i,0}$ given in $tx_2(i)$ does not match its hash commitment from $tx_1(i)$, or otherwise some commitment fails the group membership test\footnote{In~\cite{DKGRevisited,DKG} group membership of the public commitments (verifying that $X_{i,k}\in \langle g \rangle $) was omitted, but seems necessary for the proofs. We therefore add this verification check in Escrow-DKG.}, any participant $j$ can file a complaint $cm_2(j,i)$. 
            
\item[3] Sub-shares. Every participant $P_i$ locally computes her sub-shares $x_{i,j}=f_i(j)$ for all $j$. For $j \neq i$ she encrypts $x_{i,j}$ using $P_j$'s public encryption key and publishes the corresponding \emph{sub-share} transactions,
$$tx_3(i,j)=\Big[ ENC_{\gamma^{\xi_i}} \big( x_{i,j} \big) \Big]$$

If $P_i$ fails to publish $tx_3(i,j)$ in time, any participant $l$ can file a complaint, $cm_3(l,i,j)$ against her. 
            
\item[4] Sub-shares verification. Every participant $P_i$ locally verifies that the sub-shares intended to her are consistent with the public commitments. That is, she checks that for all $j$,
\begin{equation*}
g^{x_{j,i}}=\prod_{k=0}^t \big(X_{j,k}\big)^{i^k}
\end{equation*}

If the equality for some $j$ does not hold, $P_i$ can file a \textit{sub-share} complaint, $cm_4(i,j,\xi_i)$ against $P_j$.





\item[5] Conclusion. If no complaints were filed during the course of the protocol, Escrow-DKG is said to have \emph{completed successfully}. Each $P_i$ can compute her key share $x_i=\sum_{j=1}^n x_{j,i}$, as well as the global public key $X_0=\prod_{i=1}^n X_{i,0}$. She also computes the public key share $g^{x_j}$ for all participants $P_j$. Deposits are kept for the application stage. 

If some complaint is filed, it is arbitrated by the escrow, and the protocol can be relaunched with the same set of participants, excluding the one slashed by the escrow.

\end{enumerate}

\end{mdframed}
\end{minipage}
}
}

\subsubsection*{Escrow-DKG.}
Our DKG protocol, \emph{Escrow-DKG}, kick-starts the 
cryptosystem. 
Participants must first \emph{enroll} by depositing funds to the escrow. The protocol then proceeds in \emph{phases}, roughly corresponding to those of Ped-DKG. All data submitted to the escrow is publicly available to all participants. In Fig.~\ref{Escrow-DKG-fig} we depict the complete description of the protocol. 


\FloatBarrier

As in Ped-DKG, when participants detect unintended behavior they may file complaints. In Escrow-DKG, when a complaint is filed the protocol enters a dispute phase between the \emph{prover} (the entity filing the complaint and proving it) and the \emph{challenger} (the entity who challenges the complaint). For the escrow to arbitrate the complaint, the prover submits all necessary data and the escrow executes the necessary computations and rules whether the complaint is just (to accommodate arbitration to Ethereum's limitations, App.~\ref{Escrow-DKG-Ethereum} illustrates an interactive arbitration protocol in which both sides, the prover and the challenger, take active role). Then, the contract slashes and rewards accordingly. See Table~\ref{tab:complaints} for details. Once a complaint has been filed, the protocol fails. If desired, the protocol can be easily relaunched. Specifically, we do not consider failing the DKG as the end of the system's lifetime. As we discuss in Sec.~\ref{Robustness}, we expect failing the DKG to be very rare as there is no gain in doing so.

\subsubsection*{Application.}
If Escrow-DKG did not fail, we say it has completed successfully and participants can begin using their key shares, e.g., in producing threshold signatures. We emphasize the escrow keeps the deposits in order to incentivize post-DKG behavior. 

If required by the cryptosystem (as would often be), any public information generated by Escrow-DKG can be submitted to the escrow. For example, one may submit the \emph{public key}, $X_0$. Authenticity of public data submitted to the escrow is subject to complaints (see $cm_5$ in Table~\ref{tab:complaints}). However, after Escrow-DKG has completed successfully, for robustness, complaints do not fail the cryptosystem and might only invoke slashing of deposits\footnote{We note that the escrow service has no deterrence over a slashed participant, $P_j$. To handle this, the corresponding secret key share, $f(j)$, is revealed and from that point on the system operates as a $(t-1,n-1)$-threshold cryptosystem. $P_j$'s slashed deposit can be used to incentivize participants to help reveal $f(j)$.}.

The main contribution of our framework is its ability to economically discourage collusion, i.e., undesired cooperation of $t+1$ participants. We introduce the option of \emph{framing} in case collusion took place during the application stage. At any point in time, any (bonded) entity may publish a framing complaint containing evidence of collusion (framings are given in Table~\ref{tab:complaints} as $fm$). Unlike regular complaints, framing is not targeted against a specific participant but rather against the entire system. We are compelled to do so as collusion evidence cannot reveal any information regarding the (true) identities of the colluding parties. For this reason, framing results in \emph{all} deposits being burnt, except for the reward given to the framer\footnote{We do not allow framing a specific participant by reveling her private data (e.g., $f(j)$), even though it would definitely serve as means to discourage data sharing between entities. The problem is that it would also open the opportunity for anyone controlling the cryptosystem (i.e., knowing $f$) to frame honest participants, increasing the profits from collusion. As a matter of fact, ``small collusions'' (mining pools in our Bitcoin analogy) are not necessarily bad if the different entities keep having individual interests and are incentivized to frame in case the collusion became ``big''.} (see last row in Table~\ref{tab:complaints}). Framing opens a new and lucrative opportunity for entities taking part in collusion. As we show in the following section, when the system parameters are tuned correctly, entities are better off framing than colluding. Realizing this, entities are discouraged to collude, fearing another entity would frame.

\paragraph{Example I.}
In our lottery example, collusion is manifested by the generation of the signature that determines the winner \emph{during the ticket sale}. The $t+1$ colluding parties generating the signature can share the lottery prize by buying the winning ticket. On the other hand, any entity within the collusion can frame and enjoy a significant reward, while all other participants are slashed. 

\paragraph{Example II.}
In Satoshi's example, collusion is manifested by reconstructing the secret \emph{while Satoshi is still alive}. Participants should be discouraged from doing so, fearing one of them would frame the others to enjoy an additional reward.

\begingroup
\setlength{\tabcolsep}{10pt} 
\renewcommand{\arraystretch}{1.5} 

\begin{table}[h]
\caption {Complaints and Framing}
\label{tab:complaints} 

\caption*{The table depicts the details of complaints, arbitration and the outcomes in terms of slashing and rewarding. Notice that all complaints result in some value being burnt: $(n-t)\Delta$ in justified framing and $\Delta / 2$ in all other cases.}

\small
\begin{adjustwidth}{-.5in}{-.5in}  
\begin{tabular}{|p{1.5cm}|p{3cm}|p{4cm}|p{3.2cm}|}
\hline
Title & Claim & Arbitration procedure & \parbox{3cm}{\vspace{.5\baselineskip} Incentives (justified / \\ unjustified)}  \\ [0.5ex] 
\hline \hline

\multirow{2}{*}{\parbox{1.5cm}{\phantom \\ $cm_1 (j,i)$ \\or\\ $cm_3 (j,i,l)$}} & \multirow{2}{*}{\parbox{3cm}{$P_i$ did not publish $tx_1(i)$ or $tx_3(i,l)$.}}  &  \multirow{2}{*}[-\aboverulesep]{\parbox{4cm}{Trivial check.}}  & \parbox{3.2cm}{\vspace{.5\baselineskip} $P_i:-\Delta$ \\ $P_m:+\frac{\Delta}{2(n-1)}$ $\forall m\ne i$} \\
\cline{4-4}
 & & & \parbox{3.2cm}{\vspace{.5\baselineskip}$P_j:-\Delta$ \\ $P_m:+\frac{\Delta}{2(n-1)}$ $\forall m\ne j$} \\
 \hline

 \multirow{2}*{$cm_2 (j,i)$} & \multirow{2}{*}{\parbox{3cm}{$P_i$ published invalid hash commitment.}} &  \multirow{2}{*}[-\aboverulesep]{\parbox{4cm}{Compute $H(X_{i,0})$ from $tx_2(i)$ and compare it with the value in $tx_1(i)$.}}  & \parbox{3.2cm}{\vspace{.5\baselineskip}$P_i:-\Delta$ \\ $P_m:+\frac{\Delta}{2(n-1)}$ $\forall m\ne i$} \\
\cline{4-4}
 & & & \parbox{3.2cm}{\vspace{.5\baselineskip}$P_j:-\Delta$ \\ $P_m:+\frac{\Delta}{2(n-1)}$ $\forall m\ne j$} \\
 \hline

 \multirow{2}*{$cm_4 (j,i,\xi_j)$} & \multirow{2}{*}{\parbox{3cm}{$P_i$ published an inconsistent sub-share $x_{i,j}$. }} &  \multirow{2}{*}[-\aboverulesep]{\parbox{4cm}{Decrypt $tx_3(i,j)$ using $\xi_j$, and compute $g^{x_{i,j}}$. Compare with $\prod_{k=0}^t (X_{i,k})^{j^k}$ using data from $tx_2(i)$.}}  & \parbox{3.2cm}{\vspace{.5\baselineskip} $P_i:-\Delta$\\ $P_j:+\frac{\Delta}{2}$} \\
\cline{4-4}
 & & & \parbox{3.2cm}{\vspace{.5\baselineskip}$P_j:-\Delta$ \\ $P_m:+\frac{\Delta}{2(n-1)}$ $\forall m\ne j$} \\
 \hline
 
 \multirow{2}*{$cm_5 (j,i)$} & \multirow{2}{*}[-\aboverulesep]{\parbox{3cm}{$P_i$ published incorrect public data $tx_{pub}(i)$ (e.g., $g^{f(i)}$).}} &  \multirow{2}{*}[-\aboverulesep]{\parbox{4cm}{Use the commitments to compute the desired value (e.g., $\prod_{j=1}^n g^{x_{j,i}})$, and compare with $tx_{pub}(i)$.}}  &\parbox{3.2cm}{\vspace{.5\baselineskip}$P_i:-\Delta$ \\ $P_m:+\frac{\Delta}{2(n-1)}$ $\forall m\ne i$} \\
\cline{4-4}
 & & & \parbox{3.2cm}{\vspace{.5\baselineskip}$P_j:-\Delta$\\ $P_m:+\frac{\Delta}{2(n-1)}$  $\forall m\ne j$} \\
 \hline

 \multirow{2}*{$fm(j,data)$} & \multirow{2}{*}[-\aboverulesep]{\parbox{3cm}{$data$ serves as collusion evidence (e.g., $data=f(0)$).}} &  \multirow{2}{*}[-\aboverulesep]{\parbox{4cm}{Verify the evidence (e.g., check $g^{data}=\prod_{i=1}^n X_{i,0}$).}}  &\parbox{3.2cm}{\vspace{.5\baselineskip}$P_j:+t\Delta$\\ $P_m:-\Delta$ $\forall m\ne j$} \\
\cline{4-4}
 & & & \parbox{3.2cm}{\vspace{.5\baselineskip} $P_j:-\Delta$ \\ $P_m:+\frac{\Delta}{2(n-1)}$ $\forall m\ne j$ } \\
 \hline

\end{tabular}

\end{adjustwidth}

\end{table}
\endgroup

\section{Properties}
\label{properties}
The protocol satisfies three properties that we discuss hereafter. Fundamentally, we adapt the original requirements of a DKG protocol presented in Sec.~\ref{Ped-DKG} to our economic model and prove them assuming that the entities are rational.

\subsection{Correctness}

\begin{claim} [Correctness]
If Escrow-DKG (i.e., only the DKG stage) completed successfully, then the secret key shares and the public data produced by the protocol admit $(C1)$ and $(C2)$. 
\end{claim}

\noindent The proof relies on the following lemma.

\begin{lemma} \label{priftableToComplainLemma}
During Escrow-DKG, if an entity has the possibility to file a justified complaint, Escrow-DKG would not complete successfully.
\end{lemma}

\begin{proof}
We prove that if an entity has the chance to file a justified complaint against another entity, she would do so. This might seem trivial as filing a justified complaint yields an immediate profit to the prover. However, in a pessimistic scenario, since complaints require relaunching the DKG, they might reduce the value $R$---according to our model, the value $R$ is attached to a specific cryptosystem that is launched with a specific DKG run; relaunching implicitly means $R$ might change (e.g., in the lottery example, if the DKG fails, it might reduce the system's credibility and lead to a decrease in ticket sales). We show that even under the worst case assumption that $R$ reduces to zero after the complaint, it is still beneficial for entities to complain when possible (even though it implies they lose their own profit, $\frac{\alpha R}{n}$). There are three complaint categories to consider. 
\begin{enumerate}
\item Private data mismatch (sub-share complaint $cm_4$). A participant $P_j$ has clear motivation of having a secret key share that matches the public key. If one of $P_j$'s sub-shares $f_i(j)$ does not match $P_i$'s commitments, she would end up with a useless key share that would not allow her to participate in the protocol. De facto, $P_j$ would not be a participant and would lose her expected profit\footnote{We implicitly assume that the system rewards only active participants. 
}, $\frac{\alpha R}{n}$. Thus, the immediate reward of $\frac{\Delta}{2(n-1)}$ suffices for her to complain.

\item Public data mismatch ($cm_2$ and $cm_5$). In this case any entity, including external entities, can complain. Since external entities have nothing to lose and will gain a profit from complaining, they will do so\footnote{To prevent spam and unjust complaints, external complaints would be accompanied by a deposit that would be slashed if they are found unjust.}. Participants know \emph{someone} would complain and are thus incentivized to do so themselves.

\item Missing data ($cm_1$ and $cm_3$). While we do mention these complaints, in practice the escrow service itself can detect missing data. This serves enough of an incentive to complain for the immediate profit. \qed
\end{enumerate}
\end{proof}

\begin{proof} [Correctness] 
From the previous lemma we can deduce that no entity could have made a justified complaint during a successful run of Escrow-DKG. Such a run is equivalent to a Ped-DKG run in which all the participants are correct. The equivalence is given by mapping the participants in Escrow-DKG to those in Ped-DKG, and assuming they sample the same local polynomials $f_i$. If all Ped-DKG participants are correct, then no one complains and the data (public and private) produced by the protocol is identical to that in the Escrow-DKG run. $(C1)$ and $(C2)$ in Ped-DKG imply $(C1)$ and $(C2)$ in Escrow-DKG. 
\qed 
\end{proof}

Handling $(C3)$ is more tricky. Hashing $g^{a_{i,0}}$ in the enrollment transaction guarantees that when $P_j$ samples her partial secret, $a_{j,0}$, she is unaware of $g^{a_{i,0}}$ (for any $i$) and cannot relate it to those of the other participants. Assuming one participant sampled her partial secret uniformly in random is enough to conclude that the sum is also distributed uniformly. However, the option of Escrow-DKG to fail (i.e., not to complete successfully) opens the opportunity for participants to try and manipulate $g^x$. If a participant does not like the sampled $g^x$, she can fail Escrow-DKG and have it re-sampled in the subsequent run. We further relate to this issue in Sec.~\ref{Robustness}.

\subsection{Collusion resistance}
\label{collusion}


Obviously, a permissionless threshold cryptosystem has scenarios in which it cannot be trusted, e.g., when a single entity controls $t+1$ participants in the DKG. While we cannot prevent such scenarios, our goal is to articulate concrete economic conditions as to when a given cryptosystem is trustworthy and collusion resistant. We say that a threshold cryptosystem is \emph{collusion resistant} if it adheres certain conditions which make collusion a non-profitable strategy\footnote{In a collusion resistant cryptosystem, the secrecy property defined in Sec.~\ref{Ped-DKG} is satisfied. However, secrecy is not enough---collusion does not imply revealing $x$.}. The following claim gives a strong condition in this spirit.

\begin{claim} [Collusion resistance] \label{CollusionResistanceClaim}
If $R<\frac{(t-1)\Delta}{2}$ then ``two-entity'' collusion (i.e., collusion involving exactly two entities) can only take place between two entities that invested at least $\frac{(t+1) \Delta}{2}$ each. 
\end{claim}

The claim implies that one should trust the system to be collusion resistant if they \emph{believe} that there are no two entities controlling at least $\frac{t+1}{2}$ DKG participants each\footnote{The trustworthiness question we raise here is fundamental also in Bitcoin. One should trust Bitcoin only if they believe no single entity is controlling more than $50\%$ of the hash power. The open nature of Bitcoin and the ability of anyone to start mining makes this assumption reasonable.}. 





\begin{proof}
Consider the case of two separate entities $A$ and $B$ controlling $a$ and $b$ DKG participants respectively. Assume $a+b \geq t+1$ so that if the two entities collude they break the system. It is enough to show that for such collusion to be profitable, $A$ and $B$ would each have to invest at least $\frac{(t+1)\Delta}{2}$. 

We analyze the payoff matrix of $A$ and $B$ in case collusion took place 
(see Table~\ref{tab:payoff_matrix}). Both entities have two options---to frame or not to frame. 

\begin{table}[h]
\caption{Framing Payoff Matrix}
\label{tab:payoff_matrix}

\caption*{The case where both $A$ and $B$ do not frame (bottom-right) captures the pressure to collude---instead of earning their fair share ($A$ should earn $a \cdot \frac{\alpha R}{n}$), the total value held by the system, $R$, is now split solely between these two entities (according to their weights). The other cases are self explanatory from the last line in Table~\ref{tab:complaints}.}

\begin{adjustwidth}{-.15in}{-.5in}  

\begin{tabular}{|l||*{2}{c|}}

\hline
\backslashbox{Entity A}{Entity B}
&\makebox[4.5em]{Frame}&\makebox[4.5em]{Not frame}\\\hline\hline
Frame & Both have $0.5$ probability of framing first & \backslashbox{\ \ \ \ \ +$(t-a)\Delta$}{\ \ \ \ \  $-b\Delta$ \ \ \ \ \ \ } \\\hline 
Not frame &\backslashbox{\ \ \ \ \ \ \ \ \ \ \ $-a\Delta$ \ \ \ \ \ \ \ \ }{$+(t-b)\Delta$\ \ \ \ \ \ \ \ \ } & \backslashbox{$+a\cdot R/(t+1)$}{$+b\cdot R/(t+1)$}\\\hline
\end{tabular}
\end{adjustwidth}
\end{table}

In case $B$ chose to frame, $A$ is obviously better off framing as well (whoever gets her complaint processed first profits while the other gets slashed). Otherwise, in case $B$ does not frame, $A$'s utility can either be $(t-a)\Delta$ (by framing), or $\frac{R a}{a+b} \leq \frac{R a}{t+1}$ (by not framing)\footnote{Not surprisingly, the more participants an entity controls the less she is encouraged to frame.}. In case $a < \frac{t+1}{2}$, A is better off framing (for this we only need $R < (t-1)\Delta$). $B$'s view is symmetric, which means only if $a$ and $b$ both are at least $\frac{t+1}{2}$ no framing would occur.

We are not done yet. If $a < \frac{t+1}{2}$ then $B$ would disagree to collude with $A$---$B$ knows $A$'s dominant strategy is to frame. For $A$ to convince $B$ she will not frame, her payoff matrix must change---either her profit from framing reduces or her profit from a successful collusion increases. 

In the first case, $A$ would have to convince $B$ of an additional loss, $\Delta_A$, in case she (or anyone else for that matter) framed\footnote{One way to guarantee such loss would be to deposit $\Delta_A$ in a side contract which is instructed to burn the deposit in case framing took place.}. Her total investment would then be $a\Delta + \Delta_A$ and in order to be discouraged from framing, this needs to be larger than $t\Delta - \frac{Ra}{t+1}$. Simple arithmetics shows that $A$'s investment would have to be larger than $\frac{(3t-1)\Delta}{4}$ (after substituting $R$ with $\frac{(t-1)\Delta}{2}$)\footnote{We note a simple trade-off---the lower $R$'s bound is (in terms of $\Delta$), the larger $\Delta_A$ becomes. In essence this means that locking funds in side contracts is less lucrative, as more funds need to be invested in order to convince $B$ that framing is unprofitable.}. 

Since $A$'s investment is higher than $\frac{(t+1)\Delta}{2}$, we conclude that in this scenario collusion is profitable only if both sides invest more than $\frac{(t+1)\Delta}{2}$. To complete the proof we have to analyze $B$'s option to offer $A$ additional profits should the collusion succeed. 
$B$ could only offer her marginal profit from the collusion (relative to not colluding), which is $\frac{Rb}{t+1} - \frac{\alpha Rb}{n}$. This sum, combined with $A$'s profit from a successful collusion, $\frac{Ra}{t+1}$, would have to be larger than $A$'s minimal profit from framing, $(t-\frac{t+1}{2})\Delta$ (where the $\frac{(t+1)\Delta}{2}$ subtracted is derived from the fact that, in the setting suggested by the claim, $A$ is not willing to invest more than $\frac{(t+1)\Delta}{2}$). Our assumption on $R$ does not allow this even if $\alpha=0$. \qed

\end{proof}


Interestingly, collusion resistance does not depend on the parameter $\alpha$ or on the ratio $t/n$---only the ratio between $R$ and $\Delta$ matters. In practice though, the reward, $\frac{\alpha R}{n}$, needs to be attractive for participants to join. 




\begin{corollary}
In the setting of claim~\ref{CollusionResistanceClaim}, if no entity has invested at least $\frac{(t+1)\Delta}{2}$ then collusion cannot take place.
\end{corollary}

\begin{proof}
If no entity has invested $\frac{(t+1)\Delta}{2}$ then no two entities could be found to collude and obtain $t+1$ secret key shares. If entities that invested less than $\frac{(t+1)\Delta}{2}$ were to collude, forming a ``multiple-entity'' collusion, any entity within that collusion could view the payoff matrix as herself against the other entities, united. It would thus be profitable for any entity within that collusion to frame. We can thus conclude that no collusion would emerge. \qed
\end{proof}

\subsection{Robustness}
\label{Robustness}
If collusion resistance mitigates the risk of entities breaking the cryptosystem, robustness deals with the possibility of failing it. Economic robustness estimates the price of the different ways to fail the system---failing Escrow-DKG (by taking advantage of the complaint mechanism); failing the system during the application stage (by taking advantage of the framing mechanism); or putting the system to a halt (by convincing enough participants not to cooperate). 


Failing Escrow-DKG can be done by complaining unjustly or by submitting inconsistent data. We set the price of failing Escrow-DKG to be $\Delta/2$ (see Table~\ref{tab:complaints}), which is not high (but sufficient to prevent spam). This choice is sensible as not much can be gained from failing the DKG. 
Manipulating the secret $x$ is not possible, and while last actor attacks allow manipulating $g^x$, it was shown in~\cite{DKGRevisited} that the influence on $g^x$ is very limited. 
In our setting this influence comes with a price, as each attempt to re-sample $g^x$ costs $\Delta/2$.



During the application stage, framing can be used to fail the cryptosystem. The price of convincing $t+1$ participants to collude and then frame themselves must be higher than what they get in the first place which is $(t+1) \frac{\alpha R}{n}$ (additionally, they should be compensated for the burnt $\Delta$ due to the framing). Setting this price high amounts to tuning the ratio $t/n$ and the value $\alpha$ to be rather high. 


Finally, refusing to cooperate also amounts to failing the protocol. Lack of cooperation implies $n-t$ participants refuse to reveal their share. By instructing the escrow to burn all deposits in case of such lack of cooperation, we set the price of failing at $(n-t)\Delta$\footnote{Note that even though participants expect a reward of $\frac{\alpha R}{n}$ from participating 
in the cryptosystem, we make a worst case assumption where only one participant gets the whole $\alpha R$ (thus the $\frac{\alpha R}{n}$ factor is ignored in the pricing calculation).}. For this price to be high, one should set $n-t$ to be rather high.

It only makes sense to tune the parameters such that failing the protocol either way costs the same (excluding the DKG stage). This condition yields a relation between $\alpha$ and $t/n$. A simple computation gives $\alpha \sim \frac{2n(n-t)}{t^2}$. Intuitively, the smaller $\alpha$ is, the higher $t/n$ needs to be. To get $\alpha=0.25$, we need $t/n=0.9$ which sets the price of failing the application at $\frac{n}{10} \Delta$. Note the linear dependence in $n$.



\paragraph{Example I.}
\label{LotteryExplicitParameters}
Let's try some numbers in our lottery example. 
For $R=10^6$, $\Delta=10^4$ and $\alpha=0.25$, the prize is $750,000$ and $250,000$ is the reward paid to the participants generating the randomness for their services. For a $10\%$ ROI (which is lucrative for, say, a month's long investment), $n$ would have to be $250$ and $t=0.9 \cdot n=225$ (note that indeed $R$ is slightly smaller than $\frac{(t-1)\Delta}{2}$). Collusion is irrational as long as no entity has $113$ participants, which translates to a $1,130,000$ investment. Failing the protocol after the DKG would cost $25 \Delta=250,000$.


\paragraph{Example II.} 
Satoshi would deploy a smart contract\footnote{Ironically, Bitcoin is not a viable option to this end...} including the hash of her secret, and a trapdoor that allows (only) her to invoke slashing deposits. The trapdoor would expire within a year, but could be extended for another year, if Satoshi is still alive (and submits a transaction). We note that in this example, the analysis made above needs some modifications. First, it is known that the secret share holders are distinct real-world entities (the permissioned setting), and second, even in case of framing, the $R/n$ reward is guaranteed (Satoshi's signed transaction will not change). Framing can only yield an additional profit to the framers. This somewhat changes the payoff matrix in Table~\ref{tab:payoff_matrix}. In these circumstances, a less lucrative reward from framing is a better option. That is, framing would reward the framer with $\frac{t\Delta}{2}$ (instead of $t\Delta$). This would reduce the risk of $t+1$ participants deciding to split the costs of the slashing for the immediate reward $R/n$ and would enable reducing the size of the deposits (instead of $\Delta/(t+1) > R/n$, we would need $\Delta/2 > R/n$). 

Once Satoshi dies, she stops extending the trapdoor and the participants can safely cooperate and reveal the secret, as framing could no longer invoke slashing\footnote{To eliminate Satoshi's option of handing the trapdoor extension key to a friend before she dies, a hard limit should be added.}. To eliminate the risk of side contracts (mentioned in Sec.~\ref{collusion}), the deposits would have to be high enough, so that $\frac{t\Delta}{2}$ is too high an amount for the participants to put aside. 

\paragraph{Example III.}
In App.~\ref{ImproveBitcoin} we propose a multi-round leader election protocol that relies on an escrow service (which could be used to realize a sidechain to Ethereum) that mimics and improves on Bitcoin's incentive structure.



\section{Conclusion and Future Work}
\label{conclusion}
In this work we defined an economic model for rational threshold cryptosystems, which captures both the economic value produced by such systems, as well as the rationality of those taking part in the system, acting to maximize their own profits. We solve the problem of \emph{collusion pressure} inherent to rational cryptosystems by utilizing an escrow service that takes deposits from participants and uses them to discourage participants from colluding. We show how such an escrow can be realized over Ethereum. 

Our framework illustrates how threshold cryptosystems might be applicable in decentralized-permissionless settings, whereas the classic model for threshold cryptography is only applicable in trusted environments. The Satoshi example showed that our framework is suitable also in the permissioned setting, where each participant represents a distinct entity.

Our main contribution is to attach concrete economic costs to breaking the secrecy and robustness of the cryptosystem. The fact that these costs are functions of the system's parameters ($t,n,R,\alpha$) achieves two goals. First, it allows system designers to tune these parameters to suit their needs, and second, perhaps more importantly, it enables users to assess the credibility of the system. As a motivating example, we consider a set of $n$ participants that ``Shamir share'' a secret. While the participants have the ability to cooperate and reveal the secret at any time, our framework economically assures that they do so only in specific times defined by the system designers.

Our framework raises many questions. We are especially interested in the following issues, which we leave for future work. First and foremost, in the permissionless setting, collusion resistance and robustness rely on the ability to support very large $n$ and $t$ and to have an open and unlimited enrollment procedure (similarly to Bitcoin hash rate). Due to the computational (group arithmetics) and communication complexity of Ped-DKG, it can only support a limited number of participants. Future work is required to scale the number of participants---e.g., by having several DKG groups run side by side, as proposed in~\cite{Dfinity}, or use the sparse matrix approach suggested in~\cite{large-scale-DKG}.

Another interesting question is how to increase the robustness of the cryptosystem, especially during the DKG. The low cost of failing Escrow-DKG is a weak spot in our framework which could be critical in some applications. A more permissive complaint mechanism in the spirit of Ped-DKG (where complaints do not necessarily fail the DKG procedure), enhanced with slashing unjustified complaints, might be a good direction.

In the context of robustness, even if enough participants have principally decided to cooperate during the application stage (e.g., generate a threshold signature), the exact communication scheme between them is not considered in this work. In rational secret sharing~\cite{RationalSecretSharingHalpern}, communication protocols are considered with the aim that all the participants reveal the secret together. Achieving this would enhance the robustness of the cryptosystem, and its fairness.

\bibliographystyle{abbrv}
 
\bibliography{bibliography}

\clearpage 
\newpage 
\section{Appendix}
\label{appendix}

\begin{subappendices}
\renewcommand{\thesection}{\Alph{section}}%

\section{Ped-DKG and Feldman's VSS}
\label{Ped-DKG-Appendix}
The full description of Ped-DKG is given below (Fig.~\ref{Ped-DKG_Protocol}). The correctness and secrecy properties of Ped-DKG rely on a generalization of similar properties in Feldman's VSS~\cite{FeldmanVSS} (in which a single participant samples a polynomial and distributes commitments and shares, and the other participants only perform the validations).

\afterpage{
 

\begin{mdframed} [font=\small, align=left]

\captionof{figure}{\textbf{Ped-DKG}}
\label{Ped-DKG_Protocol}
\begin{enumerate}
\item Each player $P_i$ samples uniformly at random (and independently) $t+1$ coefficients from $Z_q$. These coefficients represent a random polynomial of degree $t$ over $Z_q$, denoted $f_i$:
$$f_i(z)=a_{i,0} +a_{i,1}z+\dots+a_{i,t}z^t$$
$P_i$ broadcasts her \emph{public commitments} $X_{i,k} = g^{a_{i,k}}$ for $k = 0, \dots ,t$. $P_i$ computes the \emph{sub-shares} $x_{i,j} = f_i(j)$ for $j = 1,\dots,n$ and privately sends $x_{i,j}$ to player $P_j$.

\item Each player $P_j$ verifies the sub-shares she received from the other players by checking 
\begin{equation*} \label {EqCommitment}
g^{x_{i,j}}=\prod_{k=0}^t\big(X_{i,k}\big)^{j^k}
\end{equation*}
for $i = 1,\dots,n$.

If a check fails for some index $i$, $P_j$ broadcasts a complaint against $P_i$.

\item In the case of a complaint against $P_i$ by $P_j$, $P_i$ reveals the sub-share $x_{i,j}$ (via the public broadcast communication channel). If any of the revealed sub-shares fail the above equation, $P_i$ is excluded from the set QUAL (i.e., disqualified). If more than $t$ players complain against $P_i$, she is also disqualified (as $t+1$ sub-shares reveal $f_i$ by interpolation). From the assumption on the broadcast channel, the set QUAL is defined uniquely among the correct participants.

\item The global secret $x$ itself cannot be computed by any player, but is equal to $x = \sum_{i\in \text{QUAL}} a_{i,0} = f(0)$. The public key $g^x$ can be calculated by any participant as $\prod_{i\in \text{QUAL}} X_{i,0} = g^{f(0)}$. 

Additionally, the secret key share of player $P_j$ is $x_j=f(j)=\sum _{i\in QUAL} x_{i,j}$, which only $P_j$ can compute. The corresponding public key $g^{x_j}$ however, can be calculated from the public commitments published in (1) (using Lagrange interpolation).
\end{enumerate}
\end{mdframed}
}

\begin{lemma} [Feldman]
\label{FeldmanLemma} Feldman's VSS  satisfies the following properties:
\begin{enumerate}
\item If the dealer $P_i$ is not disqualified during the protocol then all correct players hold shares that interpolate to a unique polynomial $f_i$ of degree $t$. In particular, any $t+1$ of these shares suffice to efficiently reconstruct the secret $f_i(0)$. 

\item The protocol produces information (the public values $X_{i,k}$ for $k=0,\dots,t$ and private values $x_{i,j}$ for $j=1,\dots,n$) that can be used at reconstruction time to test for the correctness of each share; thus, reconstruction is possible, even in the presence of malicious players, from any subset of shares containing at least $t+1$ correct shares. 
\item The view of the adversary is independent of the secret $f_i(0)$, and therefore secrecy is unconditional. 
\end{enumerate}
\end{lemma}

While the generalization to the Ped-DKG protocol seems straightforward, the attack described in \cite{DKG} shows there are delicate issues to address. 

\section{Threshold BLS-3 Signatures}
\label{ThresholdBLS}
Throughout the paper we describe a randomness beacon based on threshold BLS-3 signatures. For completeness, we specify the ingredients of the proposed signature scheme.

\subsubsection*{BLS Signatures.} 
In \cite{BLSPaper} the authors present a signature scheme based on the Weil pairing, which is both  simple and efficient. It results in signatures which are \emph{unique}, namely, for a pair of message and public key $(m,PK)$ there exists a unique signature $M$ which passes validation. The scheme relies on the \emph{Diffie-Hellman problem} and its variants. Let $G=\langle g \rangle$ be a cyclic group of order $q$. We consider the following problems: 

\emph{Decision Diffie-Hellman (DDH).} Let $a,b,c\in Z_q$. Given a tuple $(g,g^a,g^b,g^c)$, decide whether $g^c=g^{ab}$. 

\emph{Computational Diffie-Hellman (Co-DHP).} Let $a,b\in Z_q$. Given a tuple $(g,g^a,g^b)$, compute $g^{ab}$. 

A \emph{DDH} (resp. \emph{Co-DHP}) group is a group in which the DDH (resp. Co-DHP) is \emph{hard}. A \emph{Gap Diffie-Hellman} (GDH) group is one where the DDH is easy while the Co-DHP is hard. 

BLS utilizes GDH groups to design a signature scheme in which signing and verifying signatures is efficient: let $G=\langle g \rangle$ be a GDH group, and $x\in Z_q$ be the secret key, assumed to be known to the signer alone. Define $X=g^x$ to be the public key known to all. To sign a message $m\in G$, the signer simply raises $m$ to the power $x$, such that the signed message is $M=m^x$. Since DDH is easy in $G$, anyone can verify that $M$ is really $m^x$, by solving DDH $(g,X,m,M)$ (note that $m=g^y$ for some unknown $y\in Z_q$). On the other hand, forging a signature is infeasible since it amounts to solving the Co-DHP in $G$.

\subsubsection*{Pairings.}
One way to construct GDH groups is by using \emph{pairings}---efficiently computable maps admitting two useful properties:
\begin{definition} [Pairing]
Let $G_1,G_2$ and $G_T$ be groups. A pairing $e:G_1\times G_2\rightarrow G_T$ is a bilinear non-degenerate map.
\end{definition}

Pairings give rise to GDH groups: if $e$ is symmetric, i.e., $G_1=G_2$, and $G_1$ is a Co-DHP group, then $G_1$ is immediately a GDH group~\cite{BLSPaper}. The BLS construction uses a concrete symmetric pairing of a Co-DHP group over an \emph{elliptic curve}. This choice results in sufficient security as well as in relatively short representation. A later work on pairings~\cite{AssymetricPairings} established the fact that \emph{asymmetric pairings} can preform better in both these aspects.

\subsubsection*{Asymmetric pairings.} In case $G_1\ne G_2$ the pairing is said to be either of \emph{type-2} (if there is an efficiently computable isomorphism $\Psi: G_2\rightarrow G_1$) or of \emph{type-3} (if there is no such $\Psi$). The authors of  \cite{AssymetricPairings} define BLS-3---a variant of the BLS scheme which utilizes type-3 pairings. This allows using groups with much shorter representation, while maintaining sufficient security\footnote{In particular, a precompiled contract for computing a specific asymmetric pairing was included in the Ethereum Virtual Machine, as we elaborate in App.~\ref{Escrow-DKG-Ethereum}.}. 
The essential difference between BLS and BLS-3 is the underlying hardness assumption, which is a variant of Co-DHP, fit to the asymmetric setting:

\emph{Co-DHP\textup{*}.}  Let $G_1=\langle g_1 \rangle ,G_2=\langle g_2 \rangle $ be two cyclic groups of the same prime order $q$. Given $h,g_1,g_1^x\in G_1$ and $g_2,g_2^x\in G_2$ (where $x\in Z_q$), compute $h^x\in G_1$. 

Typically in the type-3 setting, $G_1$ computations are easier to execute~\cite{AssymetricPairings}. For this reason, messages in the BLS-3 scheme come from $G_1$---so that signatures (exponentiation of $m\in G_1$) could be computed easily. Moreover, it makes sense to have the public keys come from $G_2$, since a public key only needs to be computed once (during key generation). In the DKG however, many computations should be made in $G_2$. 

\subsubsection*{Threshold BLS signatures.} In~\cite{BoldyrevaThreshold}, Boldyreva proposes a general method for adapting various signatures schemes to the threshold setting. She specifically illustrates how to use the Ped-DKG protocol to turn BLS to a threshold signature scheme.

\section{Escrow-DKG over Ethereum}
\label{Escrow-DKG-Ethereum}
\subsection{Ethereum background}
The Ethereum blockchain, sometimes referred to as a ``world computer'', is a public blockchain in the Proof-of-Work (PoW) paradigm that maintains a global state under consensus. Digital transactions, issued by users, are collected in blocks which in turn are included to the blockchain periodically (through mining). Smart contracts are arbitrary programs, deployed to the Ethereum blockchain, and compiled to the Ethereum Virtual Machine (EVM). Transactions invoke smart contract execution and consume ``gas'' (which is a metric system to asses the amount of ``work'' operations require), in order to change the Ethereum global state.   

We highlight two main Ethereum capabilities: Ethereum as an escrow service that can follow a complex set of rules, and Ethereum as a public broadcast channel.     

\subsubsection*{Ethereum as escrow.}
A smart contract over Ethereum can easily act as an escrow service: its state allows storing a set of users along with their balances; and it can incorporate rules that determine when and how these balances are released back to the users. 
An escrow service over Ethereum has the advantage of that its corresponding escrow agent is completely abstract---there is no trusted entity that controls the escrow and can be subverted. As long as the Ethereum blockchain is viewed as a reliable system, with honest majority of hash power, the escrow cannot diverge from its prescribed set of rules.

Evidently, a sophisticated escrow service can incentivize rational users to act according to some predetermined notion of ``correct'' or desired behavior.

\subsubsection*{Communication over Ethereum.} 
The Ethereum blockchain may serve as a tamper-resistant, publicly-available and censorship-resistant ``blackboard''. 
Tamper-resistance means that once a record has been included in the blockchain (sufficiently long ago), altering it is practically impossible. Public-availability guarantees anyone has permission to read data from the blockchain. Finally, censorship-resistance means that for a reasonable fee, anyone can write arbitrary data (of reasonable size) to the Ethereum blackboard within a reasonable delay (that depends on the fee paid). These three qualities make Ethereum a great broadcast communication channel. Specifically, we assume it takes at most $\delta$ blocks to get a transaction included in a block. Additionally, by using public-key encryption (specifically, ElGamal encryption~\cite{AppliedCryptoBonehShoup}), Ethereum can be utilized as a private communication channel.

\subsection{Eth-DKG}
We now turn to specify \emph{Eth-DKG}---Escrow-DKG adapted to run efficiently over the Ethereum blockchain. The modifications essentially concern two aspects: employing specific elliptic curve groups, and minimizing on-chain resource consumption by incorporating interactive proofs for complaint arbitration. 

\subsubsection*{Precompiled contracts.} Computing general pairings and group arithmetics is costly, and executing many such computations over Ethereum would be infeasible. Fortunately, a precompiled contract that allows computing a specific asymmetric pairing $e:G_1\times G_2\rightarrow G_T$ (see~\cite{ethereum}, App. E) was incorporated into the EVM (in the Byzantium hard fork), along with several other precompiled contracts allowing to perform inexpensive $G_1$ group arithmetics.

There is still another problem to circumvent in the context of group arithmetics in Eth-DKG. Since the public key resides in $G_2$ (see App.~\ref{ThresholdBLS}), operations in Eth-DKG should be performed over $G_2$. However, the precompiled contracts only allow for efficient computations in $G_1$. To tackle this limitation, participants commit to their polynomial over \emph{both}, $G_1$ and $G_2$, by publishing $X^u_{i,k}=g_u^{a_{i,k}}$ (for $u=1,2$). We then rely on the following observation: we say that two commitments $X^1_{i,k_0}, X^2_{i,k_0}$ are \emph{$(G_1,G_2)$-consistent} if their discrete log values (with respect to $g_1$ and $g_2$) are the same. It can be  shown that, as long as $G_1$ and $G_2$ are cyclic groups, $X^1_{i,k_0}, X^2_{i,k_0}$ are $(G_1,G_2)$-consistent if and only if $e(g_1,X^2_{i,k_0})=e(X^1_{i,k_0},g_2)$ (see~\cite{PairingInversion}).

To verify that the sub-shares match the $G_2$ commitments, the smart contract would first verify that the $G_1$ commitments are consistent with the sub-shares (using the precompiled contracts for $G_1$ group arithmetics) and then that all pairs $X^1_{i,k}, X^2_{i,k}$ are $(G_1,G_2)$-consistent (using the precompiled contract for pairing computation).

\subsubsection*{Optimistic off-chain approach.}
To reduce on-chain communication to minimum, after enrollment participants are instructed to send \emph{all} public and private data via off-chain channels. In the optimistic scenario, enrollment is the only on-chain transaction. 
$P_i$'s enrollment transaction would contain additional information---$Dig_i:=Digest\big(\{X^1_{i,k}\},\{X^2_{i,k}\}\big)$, in order to make sure that she sends the same commitments $(\{X^1_{i,k}\},\{X^2_{i,k}\}\big)$ to all participants.

In order to retain the ability to arbitrate complaints over Ethereum, all off-chain communication is delivered with authentication so that the smart contract can verify the identity of the sender. This would allow using off-chain messages as evidence when complaints are filed to the contract.

\subsubsection*{Undelivered off-chain communication.}
In case $P_j$ did not receive a transaction from $P_i$ off-chain, she can submit a \emph{missing data transaction}. If the data is public in nature, $P_i$ is asked to publish it over Ethereum. If the data is private, $P_i$ encrypts it using $P_j$'s public encryption key $\gamma^{\xi_j}$ and then submits it to the smart contract (this is the reason for including $\gamma^{\xi_j}$ in the enrollment transaction). If the data is too large to be written on-chain, we propose to use the EVM memory which is much cheaper in terms of gas (but is not persistent). This could prove that indeed, some data, in the right size (but not necessarily the required data), was delivered. Then, if still needed, a regular complaint (that writes only a limited amount of data on-chain, as exemplified in App.~\ref{InetractiveComplaints} with the interactive sub-share complaint) can be filed.

Fig.~\ref{EthDKGFig} gives a full description of the Eth-DKG protocol. To keep the description concise, we assume the reader is familiar with Escrow-DKG (see Fig.~\ref{Escrow-DKG-fig}).

\afterpage{


\begin{mdframed} [font=\small, align=left]



   \captionof{figure}{\textbf{Eth-DKG}}

 

\label{EthDKGFig}
\begin{enumerate}
\item[0] Deployment. The protocol formally begins when the smart contract is deployed.
\item[1] Enrollment. The relevant Ethereum accounts enroll to the protocol by submitting \emph{enrollment} transactions \textbf{on-chain} 
$$tx_1'(i)=\Big[\Delta,\gamma^{\xi_i}, Dig_i \Big]$$

This is the only transaction (excluding complaints) in Eth-DKG which is sent on-chain.

\item[2] Pubic commitments. Every participant $P_i$ sends (off-chain) to all other participants $P_j$ her \emph{commitments} transaction
$$tx_2'(i)=[\{X^1_{i,k}\}_{k=0}^t,\{X^2_{i,k}\}_{k=0}^t]$$
If some $P_j$ received $tx_2'(i)$ that does not match $Dig_i$ she can file a complaint $cm_1'(j,i)$. If some pair of commitments $(X^1_{i,k_0}, X^2_{i,k_0})$ are not $(G_1,G_2)$-consistent, $P_j$ can file a complaint $cm_2'(j,i,k_0)$\footnote{Note that these complaints would require submitting data on-chain and might invoke interactive arbitration, as exemplified in App.~\ref{InetractiveComplaints}.}.

\item[3] Sub-shares. Every participant $P_i$ sends (over a private communication channel) to all other participants $P_j$ her corresponding \emph{sub-share} transaction
$$tx_3'(i,j)=[x_{i,j}]$$




\item[4] Sub-shares verification. Every participant $P_j$ locally verifies the sub-shares sent to her are consistent with the $G_1$ commitments she received, namely that for all $i$, 
$$g_1^{x_{i,j}}=\prod_{k=0}^t (X^1_{i,k})^{j^k}$$
If the equality for some $i'$ does not hold, she can file a complaint $cm_3'(j,i')$.

If no complaints were filed during the course of the protocol, Eth-DKG is said to have \emph{completed successfully} (the deposits are kept for the application stage). Note that the public key of the scheme is $g_2^x$, and is computed by every participant locally. If needed, anyone may publish the public key on-chain. In case of a complaint during the course of the protocol, the smart contract performs the arbitration and the protocol fails (it may be relaunched).

\end{enumerate}

\textbf{Undelivered communication.} In case $P_j$ did not receive $tx'_l(i)$ (for $l\in \{2,3\})$ in time, she can request $P_i$ to publish it on-chain by submitting a special alert to the smart contract. This alert initiates a special phase of on-chain communication. During this phase $P_i$ is required to submit the missing transaction to the smart contract---if the missing data is public ($l=2$) she submits $tx_2'(i)$ as is; otherwise, if it is private ($l=3$), she submits $tx_3''(i,j)=[ENC(x_{i,j},\gamma^{\xi_j})]$. If $P_i$ still fails to submit her on-chain transaction in time, any participant $P_m$ may file a complaint $cm_4'(m,i,l)$ against her. 
\end{mdframed}
}

\subsection{Interactive complaints over Ethereum}
\label{InetractiveComplaints}

The arbitration of certain complaints requires complex computation that does not scale well with $t$ and $n$. In this section we demonstrate a general approach to reduce on-chain 
arbitration costs via interactive protocols. Such protocols allow to shift the ``heavy-lifting'' from Ethereum to the users that are in dispute.

In high level, a communication-computation trade-off is being exploited---instead of performing an $\Omega(n)$ computation on-chain, the parties engage in an $O(\log(n))$-phase long interactive protocol that isolates ``the core of the dispute'' and consumes $O(1)$ on-chain resources per phase. We give a specific interactive protocol for $cm'_3(j,i)$, the sub-share complaint, and compare between naive arbitration and the interactive arbitration in terms of gas consumption. We emphasize that all $\Omega(n)$ complaints can be arbitrated in the spirit of this approach.


Consider the case where participant $j$ received from participant $i$ the sub-share $x_{i,j}$ and the commitments $X_{i,k}$ for $k=0, \dots ,t$ (all signed by $i$). Participant $j$ then performs the following verification:

\begin{equation} \label{equation:sub-share_check}
    g_1^{x_{i,j}}=\prod_{k=0}^t\big(X_{i,k}\big)^{j^k}
\end{equation}

\noindent In the case of failure, $j$ submits a complaint transaction against $i$: $cm'_3(j,i)$.
This initiates a dispute stage between $j$ (the “prover”) and $i$ (the “challenger”), where  the contract takes the role of the arbitrator.

Settling the dispute naively would require $j$ to submit to the contract both, $i$’s commitments and the sub-share $x_{i,j}$. Then, the contract would verify equation~\ref{equation:sub-share_check} and determine whether the complaint was just. However, this approach consumes storage and computation linearly in $t$ and would be unsustainable over Ethereum. Moreover, we put special attention on \emph{elliptic curve} operations~\cite{ellipticCurveGuide} which are extremely expensive---while a sustainable arbitration procedure might consume $O(log(t))$ on-chain resources, it should however reduce elliptic curve computations to minimum. We propose the interactive approach in Fig.~\ref{Inter-Sub-share_complaint} which satisfies the above requirements (in particular, it requires only $O(1)$ elliptic curve computations). 

\afterpage{

\begin{mdframed} [
	font				= \small, 
	align				= left,
	innertopmargin		= 4pt,
	innerbottommargin	= 4pt,
	innerrightmargin	= 4pt,
	innerleftmargin		= 4pt,
	skipabove			= 10pt,
	skipbelow			= 10pt,
	nobreak				= true
]

\captionof{figure}{\textbf{Interactive sub-share complaint}}
\label{Inter-Sub-share_complaint}
Let $j$ be the prover and $i$ be the challenger.

    \begin{enumerate}
    
        \item As a preparatory step each of the parties computes locally:
        
        $$ \zeta (m)=\prod_{k=0}^m\big(X_{i,k}\big)^{j^k} \text{ for } m = 0,\dots,t$$
        
        We denote by $\zeta^i(*)$ and $\zeta^j(*)$ the values $i$ and $j$ compute respectively. 
        
        \item The dispute phase progresses in rounds, where in each round, $i$ and $j$ interact by submitting transactions to the contract in turns\footnote{If one of them, in her turn, fails to submit a valid transaction within a predetermined time (block height), she forfeits the dispute and loses her deposit.}. The contract's state is initialized: $l \gets -1$, $ h \gets t+1$. 
        
        \textbf{While} $h-l > 1$ execute the following round:
        
            \begin{enumerate}\item[]
            \begin{enumerate}[label=(\alph*)]

                \item Each of the parties locally compute $m = l + \big\lceil \frac{h-l}{2} \big\rceil$.
    
               \item $i$ submits to the contract $\zeta^i(m)$. 
                The contract updates: $last \gets \zeta^i(m)$.
                
                \item $j$ submits to the contract the bit $\beta$ where:
                $$ 
                    \beta =	
                        \begin{cases}
                    		  `agree`  &\mbox{if } \zeta^i(m) = \zeta^j(m) \\
                     		  `disagree`  &\mbox{otherwise}
                         \end{cases}
                $$
                If $\beta=`agree`$, the contract updates: 
                $highest_{agree} \gets last \text{, } l \gets m$.
                
                Otherwise, the contract updates:  
                $lowest_{disagree} \gets last \text{, } h \gets m$.
                
                
            \end{enumerate}\end{enumerate}
        \item Once the stopping condition is met, $h=l+1$, and there are three possible cases:
            \begin{enumerate}
                \item If $l > -1$ and $h < t+1$, 
                 $j$ submits to the contract $X_{i,l+1}$, signed by $i$. The contract verifies:
                    $$highest_{agree} \cdot (X_{i,l+1})^{j^{l+1}} = lowest_{disagree}$$

                \item If $l=-1$, $j$ submits to the contract $X_{i,0}$, signed by $i$. The contract verifies:
                    $$lowest_{disagree} = X_{i,0}$$

                \item If $h=t+1$, $j$ submits to the contract $x_{i,j}$, signed by $i$. The contract verifies:
                    $$highest_{agree} = g_1^{x_{i,j}}$$
               
               \item[] For each of the above cases, if the equality holds the complaint is ruled unjust, otherwise it is ruled just.
            \end{enumerate}
        
    \end{enumerate}
\end{mdframed}
}




\begin{claim}
Given that equation~\ref{equation:sub-share_check} does not hold and that the prover $j$ follows her instructions in Fig.~\ref{Inter-Sub-share_complaint}, then her complaint is ruled just. Conversely, if the equation does hold and the challenger $i$ follows her instructions in Fig.~\ref{Inter-Sub-share_complaint}, then $j$'s complaint is ruled unjust.
\end{claim}

Before heading to the proof, we note that the repetitive procedure in step $2$ employs a binary search between the parties to find where they first disagree in the right-hand side computation of equation~\ref{equation:sub-share_check}. By the end of step two, the parties' dispute is reduced to one of the following cases:
\begin{enumerate}
\item Case 3.(a) illustrates the scenario in which the parties agree on some $\zeta(l)$ but disagree on $\zeta(l+1)$ for $0 \leq l \leq t-1$. Then, the contract calculates $\zeta(l+1)$ from $\zeta(l)$ and $X_{i,l+1}$ (where the latter is supplied to the contract by $j$ and must be signed by $i$). If the value $\zeta(l+1)$ that the contract calculated conflicts with $\zeta^i(l+1)$ then the complaint is ruled just. Otherwise, it is ruled unjust.
    
\item Case 3.(b) illustrates the scenario in which the parties disagree on $\zeta(0)$. Then, if $j$ supplies the contract with $X_{i,0}$, signed by $i$, which conflicts with $\zeta^i(0)$ then the complaint is ruled just. Otherwise, it is ruled unjust.

\item Case 3.(c) illustrates the scenario in which the parties agree on $\zeta(t)$. Then, if $j$ supplies the contract with $x_{i,j}$, signed by $i$, where $g_1^{x_{i,j}}$ conflicts with $\zeta^i(t)$ then the complaint is ruled just. Otherwise, it is ruled unjust. 
\end{enumerate}

\begin{proof}
First assume equation~\ref{equation:sub-share_check} does hold but $j$ complains. We show that the complaint is ruled unjust. During step $2$, $j$ must admit one of the following strategies:
\begin{enumerate}
\item If $j$ at some point agrees and at another point disagrees, then we are in case 3.(a), and from the assumption that $i$ follows her instructions correctly, $\zeta^i(l)=\prod_{k=0}^l\big(X_{i,k}\big)^{j^k}$, $j$ must supply the contract with $X_{i,l+1}$ (from the assumption that equation~\ref{equation:sub-share_check} holds). Indeed, since $\zeta^i(l+1)=\prod_{k=0}^{l+1}\big(X_{i,k}\big)^{j^k}$ the contract would compute $\zeta(l+1)=\zeta^i(l+1)$ and rule that the complaint is unjust.

\item If $j$ always disagrees, then we are in case 3.(b) and from the assumption that $i$ follows her instructions correctly, $\zeta^i(0)=X_{i,0}$. $j$ must supply the contract with $X_{i,0}$ (from the assumption that equation~\ref{equation:sub-share_check} holds). Indeed, since $\zeta^i(0)$ does not conflict with $X_{i,0}$ the contract rules that the complaint is unjust.

\item If $j$ always agrees, then we are in case 3.(c) and from the assumption that $i$ follows her instructions correctly and that equation~\ref{equation:sub-share_check} holds, $\zeta^i(t)=g_1^{x_{i,j}}$. $j$ must supply the contract with $x_{i,j}$ and the contract computes $g_1^{x_{i,j}}$. Indeed, since $\zeta^i(t)$ does not conflict with $g_1^{x_{i,j}}$ the contract rules the complaint unjust.
\end{enumerate}

In the other direction, we show that if equation~\ref{equation:sub-share_check} does not hold and $j$ complains, then the complaint is ruled just. During step $2$, $i$ must admit one of the following strategies:
\begin{enumerate}
\item If $i$ always submits $\zeta(m)$ then $j$ will always agree and step $2$ ends in case 3.(c). From the assumption that the equation~\ref{equation:sub-share_check} does not hold, once $j$ submits $x_{i,j}$ and the contract computes $g_1^{x_{i,j}}$, it finds a mismatch with $\zeta(t)$ and rules that the complaint is just.    

\item If $i$ at some point submits $\zeta^i(m) \ne \zeta(m)$ then step $2$ ends in case 3.(a) or 3.(b). If $i$ always submits $\zeta^i(m) \ne \zeta(m)$ then eventually, she submits $\zeta^i(0) \ne X_{i,0}$. When $j$ submits $X_{i,0}$ the contract rules the complaint as just. Otherwise, If $i$ submits some $\zeta^i(m) = \zeta(m)$ then the interactive protocol would find some $0\leq m^* \leq t-1$ such that $\zeta^i(m^*) = \zeta(m^*)$ and $\zeta^i(m^*+1) \ne \zeta(m^*+1)$. When $j$ submits $X_{i,m^*+1}$ the contract rules the complaint as just.
\end{enumerate} 
\qed
\end{proof}


We turn to analyze the computational complexity of the interactive dispute protocol. As mentioned before, the repetitive procedure in step $2$ performs a binary search over $t+1$ values thus the dispute takes $O(\log(t))$ rounds. At every round each of the parties submits to the contract a constant amount of data and the contract performs a constant amount of computations. Together with the final step, where constant computations are performed, we conclude that the overall complexity of the dispute phase is $O(\log(t))$ in time and space. 
We further emphasize that elliptic curve computations occur only in step $3$, hence our protocol requires only $O(1)$ such computations in total.

In Fig.~\ref{fig:subshare_complaint_gas} we show the actual gas costs of the two approaches to arbitrate the sub-share complaint---the interactive approach and the naive on-chain approach. The results are taken from a prototype implementation of the two arbitration procedures.  

We note that currently\footnote{As for September 2018.} Ethereum's block gas limit is approximately $8$M gas. Hence the naive approach would be practically infeasible to run in a smart contract for any threshold values that are greater than $t=150$. In contrast, the interactive approach can handle very high threshold values, e.g., for $t=100,000$ the gas cost for each of the parties is much lower than $1$M gas.

\begin{figure}[t!]
\begin{tikzpicture}
  \begin{axis}
  [ 
  	xmode		= log,
    log ticks with fixed point,
    xlabel		= {threshold $t$},
    ylabel		= {gas cost},
    xmin		= 0, 
    xmax		= 100000,
    ymin		= 0, 
    ymax		= 1000000,
    xtick		= {1,10,100,1000,10000,100000},
    ytick		= {0,200000,400000,600000,800000,1000000},
    height 		= 0.42 \textwidth,
    width 		= 0.95  \textwidth, 
    grid        = major, 
    legend style= {at={(0.99,0.38)},font=\fontsize{7}{5}\selectfont},
    x tick label style={
    	font=\footnotesize,
        /pgf/number format/.cd,
            fixed,
            precision=2,
        /tikz/.cd
    },
    y tick label style={font=\footnotesize},
    label style={font=\footnotesize}
  ] 

       \addplot[color=red, mark=triangle*, only marks] coordinates {(5,  328532)(20,  407624)(50,  447170)(100,  486716)
    (500,  565808)(1000,  605354)(10000,  763538)
    (100000,  842630)
    };
    \addlegendentry{prover (interactive)};
    \addplot[color=blue, only marks] coordinates {
    (5,  139575)(20,  212625)(50,  249150)(100,  285675)
    (500,  358725)(1000,  395250)(10000,  541350)
    (100000,  614400)
    };
    \addlegendentry{challenger (interactive)}; 
    \addplot[color=black, mark=square*, only marks] coordinates {
    (1,  237232)(5,  451782)(10,  719973)(15,  988170)
    };
    \addlegendentry{prover (naive)}; 
  \end{axis}
\end{tikzpicture}
\caption{Sub-share complaint arbitration gas costs. Two approaches are presented: a single round (naive) arbitration; an interactive multi-round arbitration.}
\label{fig:subshare_complaint_gas}
\end{figure}
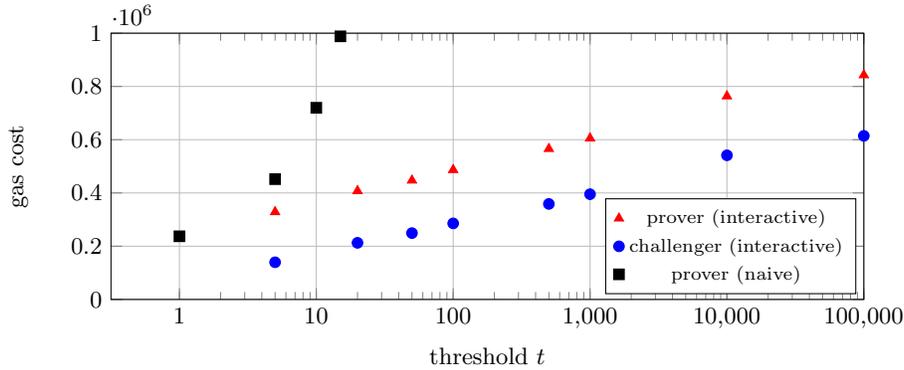




\section{Incentive driven multi-round leader election}
\label{ImproveBitcoin}
In Bitcoin, miners invest money in mining hardware, electricity and maintenance in order to participate in an ongoing PoW lottery. When elected, miners have the right to mint new bitcoins. While the mining game has proved to attract many miners, it suffers two main problems. First, it was shown that miners, following selfish mining and block withholding strategies~\cite{SelfishMining,SelfishMining-AvivZohar}, could bias the lottery to their benefit, if they control enough of the hash power. This incentivizes miners to collude and form large mining pools. In a situation where $51\%$ of the mining power is in the hands of colluding miners, they can ensure $100\%$ of the rewards to themselves. The second problem is that mining is not a "fair" game. Ideally, a miner's probability of winning the lottery would be proportional to the money they invested. With variations in mining hardware, electricity costs and software advancements like AsicBoost~\cite{AsicBoost}, the mining landscape in Bitcoin is far from ideal. 

We propose a multi-round leader election protocol, in which the leader in round $r$ is elected according to a secret value $S^r$, which is kept hidden until round $r$ and is revealed only at the beginning of the round. Similarly to our lottery example, $S^r$ is the unique threshold signature of $S^{r-1}$. The leader is elected among a set of candidates, which are also the ones producing the signature. At the start of the protocol they all engage in Escrow-DKG to generate the appropriate keys (recall that in Escrow-DKG participants are required to make a deposit to the escrow). 




Similarly to Bitcoin, by publishing $S^r$ the elected leader is entitled to mint $\rho$ new coins, where in our model notations, $R=\rho \cdot e$, and $e$ is the total number of rounds in the protocol. Once $e$ rounds have passed, Escrow-DKG is relaunched with a fresh set of participants, and possibly larger deposits. This implies that Escrow-DKG would be run over Ethereum periodically, say every week or month. It is crucial for Escrow-DKG to support a very large number of participants, $n$ (so that it is unfeasible for one entity to control too many DKG participants). 
The $S^r$s would be included in Ethereum blocks according to some predetermined schedule (say every $40^\text{th}$ block). Once a new $S^r$ has been published, the elected leader would propose a sidechain block, hash-referencing the Ethereum block that contains $S^r$. Obviously, the leader would also hash-chain the previous block in the sidechain. 

While chain splits are definitely possible, at this point we will not elaborate as to how they might be resolved. We will say that chain splits can be addressed by following the longest chain rule, or other chain selection rules dictated by Byzantine agreement protocols (e.g., PBFT~\cite{PBFT}, Tendermint~\cite{Tendermint}, Casper FFG~\cite{CasperFFG}, or Hot-stuff~\cite{Hot-Stuff}). We shall also say that in order to incentivize participants to follow the chain selection rule, its logic should be enforced by the Ethereum smart contract that serves as escrow\footnote{Long-range attacks could be addressed by the ``social consensus'' approach~\cite{weak-subjectivity}.}.


Chain rules aside, $t+1$ colluding parties would be able to reveal the $S^r$s in advance and could translate this information to disproportionate rewards. This is also true in Bitcoin as mentioned earlier. There are a few advantages to our approach though. First, $t$ can be tuned. Specifically, $t$ could be larger than $n/2$ resulting in better security. Of course, the threshold parameter must coincide with the chain selection rule and tuning it to be too large puts the robustness of the sidechain in risk (this trade-off has been discussed in Sec.~\ref{properties}). Second, framing, which is actualized by submitting $S^r$ prior to its prescribed Ethereum height, makes the system trustworthy as long as no entity is controlling $\frac{t+1}{2}$ DKG participants. This would be a realistic assumption if anyone that wished to do so could enrol to Escrow-DKG, and the total deposit, $\Delta \cdot n$, would be high enough (relative to $R$). 

Unlike in Bitcoin, in our proposal, the probability of each DKG participant to be elected leader is equal. The only way to increase one's probability, and thus also her fraction of rewards, is to buy more DKG participants. Our protocol induces a fair lottery---with a linear relationship between the sum invested and the expectancy of rewards. Moreover, breaking this linear relationship comes with a known and clear price tag\footnote{Without elaborating, the option to frame hidden chains could largely mitigate the risk of long-range attacks.}.  


This scheme relies on a super-scalable DKG and the need to reconstruct threshold signatures with very large $t$s---this would be computationally intensive and would require significant communication among the participants. One approach to circumvent this problem could borrow from Dfinity~\cite{Dfinity}---use many parallel independent DKG protocols run side by side. 
Another approach, presented in~\cite{large-scale-DKG}, generalizes on the traditional SSS-based DKG approach with independent linear equations dictated by a matrix. Using adequate sparse matrices yields in highly a scalable DKG procedure (in which reconstruction has a negligible probability of failing, even if enough participants try to reconstruct).




\end{subappendices}


\end{document}